\documentclass[11pt]{article}
\usepackage{epsfig}
\usepackage{graphicx}
\usepackage{setspace}
\usepackage{color}
\usepackage{amsmath,amsthm,amssymb}
\usepackage{latexsym}
\newtheorem{theor}{Theorem}[section]
\newtheorem{theo}[theor]{Theorem}

\newtheorem{lemma}[theor]{Lemma}
\newtheorem{defi}{Definition}[section]
\usepackage{hhline}

  \baselineskip 8 mm

\begin{document}

\title{ \bf {On the Cubicity of Bipartite Graphs} }

\author{L. Sunil Chandran$^*$, Anita Das$^\dag$, Naveen
Sivadasan$^\ddag$ \footnote {This research was funded by the DST grant SR/S3/EECE/62/2006}\\\\
$^{*, \ \dag}$ Computer Science and Automation department, Indian
Institute of
Science\\
Bangalore- 560012, India.\\
\{sunil, anita\}@csa.iisc.ernet.in\\
$\ddag$ Advanced Technology Center, Tata Consultancy Services\\
1, Software Units Layout, Madhapur\\
Hyderabad - 500081, India\\
s.naveen@atc.tcs.com}


\date{}
\maketitle

\begin{abstract}
{\it A unit cube in $k$-dimension (or a $k$-cube) is defined as the
cartesian product $R_1 \times R_2 \times \cdots \times R_k$, where
each $R_i$ is a closed interval on the real line of the form $[a_i,
a_i+1]$. The {\it cubicity} of $G$, denoted as $cub(G)$, is the
minimum $k$ such that $G$ is the intersection graph of a collection
of $k$-cubes. Many NP-complete graph problems can be solved
efficiently or have good approximation ratios in graphs of low
cubicity. In most of these cases the first step is to get a low
dimensional cube representation of the given graph.

It is known that for a graph $G$, $cub(G) \leq
\left\lfloor\frac{2n}{3}\right\rfloor$. Recently it has been shown
that for a graph $G$, $cub(G) \leq 4(\Delta + 1)\ln n$, where $n$
and $\Delta$ are the number of vertices and maximum degree of $G$,
respectively. In this paper, we show that for a bipartite graph $G =
(A \cup B, E)$ with $|A| = n_1$, $|B| = n_2$, $n_1 \leq n_2$, and
$\Delta' = \min\{\Delta_A, \Delta_B\}$, where $\Delta_A =
\mbox{max}_{a \in A}d(a)$ and $\Delta_B = \mbox{max}_{b \in B}d(b)$,
$d(a)$ and $d(b)$ being the degree of $a$ and $b$ in $G$
respectively, $cub(G) \leq 2(\Delta'+2) \lceil \ln n_2 \rceil$. We
also give an efficient randomized algorithm to construct the cube
representation of $G$ in $3(\Delta'+2)\lceil \ln n_2 \rceil$
dimensions. The reader may note that in
general $\Delta'$ can be much smaller than $\Delta$.}\\


\end{abstract}

\section{Introduction}

Let $\mathcal{F}$ be a family of non-empty sets. An undirected graph
$G$ is an intersection graph for $\mathcal{F}$ if there exists a
one-one correspondence between the vertices of $G$ and the sets in
$\mathcal{F}$ such that two vertices in $G$ are adjacent if and only
if the corresponding sets have non-empty intersection. If
$\mathcal{F}$ is a family of intervals on real line, then $G$ is
called an {\it interval graph}. If $\mathcal{F}$ is a family of
intervals on real line such that all the intervals are of equal
length, then $G$ is called a {\it unit interval graph}.


A unit cube in $k$-dimensional space or a $k$-cube is defined as the
cartesian product $R_1 \times R_2 \times \cdots \times R_k$, where
each $R_i$ is a closed interval on the real line of the form $[a_i,
a_{i}+1]$. A $k$-cube representation of a graph is a mapping of the
vertices of $G$ to $k$-cubes such that two vertices in $G$ are
adjacent if and only if their corresponding $k$-cubes have a
non-empty intersection. The {\it cubicity} of $G$ is the minimum $k$
such that $G$ has a $k$-cube representation. Note that a $k$-cube
representation of $G$ using cubes with unit side length is
equivalent to a $k$-cube representation where the cubes have side
length $c$ for some fixed positive number $c$. The graphs of
cubicity $1$ are exactly the class of unit interval graphs.

The concept of cubicity was introduced by F. S. Roberts
\cite{roberts69} in 1969. This concept generalizes the concept of
unit interval graphs. If we require that each vertex of $G$
correspond to a $k$-dimensional axis-parallel box $R_1 \times R_2
\times \cdots \times R_k$, where each $R_i$, $1 \leq i \leq k$, is a
closed interval of the form $[a_i, b_i]$ on the real line, then the
minimum dimension required to represent $G$ is called its {\it
boxicity} denoted as $box(G)$. Clearly $box(G) \leq cub(G)$, for a
graph $G$. It has been shown that deciding whether the cubicity of a
given graph is at least three is NP-complete \cite{yannakakis}.
Computing the boxicity of a graph was shown to be NP-hard by Cozzens
in \cite{cozzens81}. This was later strengthened by Yannakakis
\cite{yannakakis}, and finally by Kratochvil \cite{kratochvil94} who
showed that deciding whether boxicity of a graph is at most two
itself is NP-complete.


Thus, it is interesting to design efficient algorithms to represent
small cubicity graphs in low dimension. There have been many
attempts to bound the cubicity of graph classes with special
structure. The cube and box representations of special classes of
graphs like hypercubes and complete multipartite graphs were
investigated in \cite{irith, cub1, cub2, maehara, quint, roberts69,
trotter}.


\subsection{Our results}

Recently Chandran {\sl et al.} \cite{cubicity} have shown that for a
graph $G$, $cub(G) \leq 4(\Delta + 1)\ln n$, where $n$ and $\Delta$
are the number of vertices and maximum degree of $G$, respectively.
In this paper, we present an efficient randomized algorithm to
construct a cube representation of {\it bipartite graphs} in low
dimension. In particular, we show that for a bipartite graph $G = (A
\cup B, E)$, $cub(G) \leq 2(\Delta'+2) \lceil \ln n_2 \rceil$, where
$|A| = n_1$, $|B| = n_2$, $n_1 \leq n_2$, and $\Delta' =
\min\{\Delta_A, \Delta_B\}$, where $\Delta_A = \mbox{max}_{a \in
A}d(a)$ and $\Delta_B = \mbox{max}_{b \in B}d(b)$, $d(a)$ and $d(b)$
being the degree of $a$ and $b$ in $G$, respectively. The algorithm
presented in this paper is not very different from that of
\cite{cubicity} but this has the advantage that it gives  a better
result in the case of bipartite graphs. Note that, $\Delta'$ can be
much smaller than $\Delta$ in general, where $\Delta$ is the maximum
degree of $G$. In particular, when $|A| \ll |B|$, then the bound for
cubicity given in this paper can be much better than that given in
\cite{cubicity}. Also, the complexity of our algorithm is comparable
with the complexity of the algorithm proposed in \cite{cubicity}.

\section{Preliminaries}

Let $G = (A \cup B, E)$ be a simple, finite bipartite graph. Let
$|A| = n_1$, $|B| = n_2$, and $n_1 \leq n_2$. Let $N(v)= \{w \in
V(G) | vw \in E(G)\}$ be the set of neighbors of $v$. Degree of a
vertex $v$, denoted as $d(v)$, is defined as the number of edges
incident on $v$. That is, $d(v) = |N(v)|$. Suppose $\Delta_A$ denote
the maximum degree in $A$ and $\Delta_B$ denote the maximum degree
in $B$. That is, $\Delta_A = \mbox{max}_{a \in A}{d(a)}$ and
$\Delta_B = \mbox{max}_{b \in B}{d(b)}$.

For a graph $G$, let $G'$ be a graph such that $V(G') = V(G)$. Then,
$G'$ is a {\it super graph} of $G$ if $E(G) \subseteq E(G')$. We
define the {\it intersection} of two graphs as follows: if $G_1$ and
$G_2$ are two graphs such that $V(G_1) = V(G_2)$, then the
intersection of $G_1$ and $G_2$ denoted as $G = G_1 \cap G_2$ is a
graph with $V(G) = V(G_1) = V(G_2)$ and $E(G) = E(G_1) \cap E(G_2)$.

Let $I_1, I_2, \ldots, I_k$ be $k$ unit interval graphs such that $G
= I_1 \cap I_2 \cap \cdots \cap I_k$, then $I_1, I_2, \ldots, I_k$
is called an {\it unit interval graph representation} of $G$. The
following equivalence is well known.

\begin{theo} [\cite{roberts69}] \label{1}
The minimum $k$ such that there exists a unit interval graph
representation of $G$ using $k$ unit interval graphs $I_1, I_2,
\ldots, I_k$ is the same as $cub(G)$.
\end{theo}

\section{Construction}

Let $G = (A \cup B, E)$ be a bipartite graph. In this section we
describe an algorithm to efficiently compute a cube representation
of $G$ in $2(\Delta'+2) \lceil \ln n_2 \rceil $ dimensions, where $\Delta' =
\min\{\Delta_A, \Delta_B\}$.

\begin{defi}
Let $\pi$ be a permutation of the set $\{1, 2,\ldots,n\}$ and $X
\subseteq \{1, 2, \ldots, n\}$. The projection of $\pi$ onto $X$
denoted as $\pi_X$ is defined as follows. Let $X = \{u_1,
u_2,\ldots, u_r\}$ be such that $\pi(u_1) < \pi(u_2) < \ldots <
\pi(u_r)$. Then $\pi_X(u_1) = 1, \pi_X(u_2) = 2, \ldots, \pi_X(u_r)
= r$.
\end{defi}

\begin{defi}
A graph $G = (V, E)$ is a unit interval graph if and only if there
exists a function $f : V \longrightarrow R$ and a constant $c$ such
that $(u, v) \in E(G)$ if and only if $|f(u) - f(v)| \leq c$.
\end{defi}

\noindent{\bf Remark:} Note that the above definition is consistent
with the definition of the unit interval graphs given at the
beginning of the introduction.\\

Let $G = (A \cup B, E)$ be a bipartite graph.
Given a permutation of the vertices of $A$, we construct
a unit interval graph $U(\pi, A, B, G)$ as follows. Let $f : A \cup
B \longrightarrow R$ be such that if $v \in A$, then $f(v) = \pi(v)$
and if $v \in B$, then $f(v) = n + \min_{x \in N(v)}{\pi(x)}$. Two
vertices $u, v \in A \cup B$ are made adjacent if and only if $|f(u) -
f(v)| \leq n$, where $n = |A| + |B| = n_1 + n_2$.\\

\noindent{\bf Claim 1:} Let $G' = U(\pi, A, B, G)$. Then $G'$ is a
supergraph of $G$.

\begin{proof}

Suppose $(a, b) \in E(G)$. Without loss of generality suppose
$a \in A$ and $b \in B$. Let $s = \min_{x \in N(b)}{\pi(x)}$. So,
$f(b) = n + s$. As $f(a) = \pi(a)$ and $a \in N(b)$, $\pi(a) \geq
s$. Therefore we have, $|f(b) - f(a)| = n + s - \pi(a) \leq n$. Thus
$(a, b) \in E(G')$. Hence $G'$ is a supergraph of $G$.
\end{proof}

\noindent {\bf Remark:} Note that if we reverse the roles of $A$ and
$B$ in the above construction, i.e., if we start with a permutation
of the vertices of $B$ rather than that of $A$, then the resulting
unit interval graph will be denoted as $U(\pi,B,A,G)$. Clearly,
$U(\pi,B,A,G)$ will also be a super graph of $G$.\\

\noindent{\bf RANDUNIT}

Input: A bipartite graph $G = (A \cup B, E)$.

Output: A unit interval graph $G'$ which is a supergraph of $G$.

\noindent{\bf begin}

{\bf if} ($\Delta_B \leq \Delta_A$) then

Step 1. Generate a permutation $\pi$ of $\{1, 2,\ldots, n_1\}$ (the
vertices of $A$)

uniformly at random.

Step 2. Return $G' = U(\pi, A, B, G)$.

{\bf else}

Step 1. Generate a permutation $\pi$ of $\{1, 2,\ldots, n_2\}$ (the
vertices of $B$)

uniformly at random.

Step 2. Return $G' = U(\pi, B, A, G)$.

\noindent{\bf end}\\

\begin{lemma}\label{1.1}
Let $a \in A$ and $b \in B$ be such that $e = (a, b) \notin E(G)$.
Let $G'$ be the output of {\bf RANDUNIT($G$)}. Then

\begin{center}
$\Pr[e \in E(G')] \leq \frac{\Delta'}{\Delta' + 1}$.
\end{center}
\end{lemma}

\begin{proof}
{\bf Case I:} $\Delta_B \leq \Delta_A$.

Let $\pi$ be a  permutation of the vertices in $A$. Let $G' = U(\pi,
A, B, G)$. Suppose two vertices $a \in A$ and $b \in B$ are
non-adjacent in $G$. Let $t = \min_{x \in N(b)}{\pi(x)}$.\\

\noindent{\bf Claim:} The vertices $a$ and $b$ will be adjacent in
$G'$ if and only if $\pi(a) > t$.

If $a$ and $b$ are adjacent in $G'$, then we have $|f(b) - f(a)| =
|(n+t) - \pi(a)| \leq n$, i.e., $\pi(a) > t$. Hence $a$ is adjacent
to $b$ in $G'$.

So, $\Pr[e \in E(G')] = \Pr[\pi(a) > t]$ = $1 - \Pr[\pi(a) < t]$.
(Note that $\pi(a) \ne t$, since $a \notin N(b)$.)
Let $X = \{a\} \cup N(b)$ and $\pi_X$ be the projection of $\pi$ on
$X$. Total number of permutations of $X$ is $(d(b) + 1) !$. Now, it
can be easily seen that $\pi(a) < t$ if and only if $\pi_X(a) = 1$.
Thus,

\vspace{-.7cm}

\begin{eqnarray*}
  \Pr[(a, b) \in E(G')]
     &=& 1 - \frac{d(b) !}{(d(b) + 1) !} \\
     &=& \ \frac{d(b)}{d(b) + 1} \\
     &\leq& \frac{\Delta'}{\Delta' + 1}
     \end{eqnarray*}

Hence the lemma.

\noindent{\bf Case II:} $\Delta_A \leq \Delta_B$.

Let $\pi$ be the permutation of the vertices in $B$. Let $G' = U(\pi,
B, A, G)$. Proof is similar to case I.
\end{proof}

\begin{lemma}\label{3.2}
Given a bipartite graph $G = (A \cup B, E)$,  there exists a super
graph $G^*$ of $G$ with $cub(G^*) \le 2(\Delta' + 1) \ln n_2$, such
that if $u \in A$, $v \in B$ and $(u,v) \notin E(G)$, then $(u,v)
\notin E(G^*)$.
\end{lemma}

\begin{proof}
Let $U_1, U_2,\ldots,U_t$ be the unit interval graphs generated by $t$
invocations of {\bf RANDUNIT($G$)}. Clearly $U_i$, for each $i$, $1
\leq i \leq t$, is a super graph of $G$ by Claim 1. Let $G^* = U_1
\cap U_2 \cap \cdots \cap U_t$.  Now let $u \in A$, $v \in B$ and
$(u,v) \notin E(G)$. Then,

$\Pr[(u, v) \in G^*] = \Pr\left[\displaystyle\bigwedge_{1 \leq i
\leq t}{(u, v) \in E(U_i)}\right] \leq \left (\frac{\Delta'}{\Delta' + 1} \right )^t$
(Applying  Lemma \ref{1.1}).
Now,

\begin{eqnarray*}
   Pr\left[\bigvee_{u \in A, b\in B, (u, v) \notin E(G)} {(u, v) \in E(G^*)}\right]
   &<& n_1n_2 \left(\frac{\Delta'}{\Delta' + 1}\right)^t \\
   &\le& n_2^2 \left(1 - \frac{1}{\Delta' + 1}\right)^t \\
   &\le& n_2^2 \ e^{-\frac{t}{\Delta' + 1}}
\end{eqnarray*}

If  $t = 2(\Delta' + 1)\ln n_2$ the above probability is $< 1$. Thus
we infer that there exists a super graph $G^*$ of $G$ such that if
$u \in A$, $v\in B$ and $(u,v) \notin E(G)$, $(u,v) \notin E(G^*)$ also.
From the definition of $G^*$ we have $cub(G^*) \le 2 (\Delta' + 1) \ln n_2$. Hence the
Lemma.
\end {proof}

\noindent {\bf Remark:} If we had chosen  $t = 3(\Delta' + 1)\ln
n_2$ in the above  proof, we can substantially reduce the failure
probability. More precisely we can get $$\Pr (G^* \mbox { does not
satisfy the desired property } )  \leq \frac{1}{n_2}$$

Now we will construct two special graphs $H_1$ and $H_2$ such that
$H_i$ is a super graph of $G$ for $i=1,2$.

\begin {defi}
\label {H1definition}
Let $A = \{v_1, v_2, \ldots,v_{n_1}\}$.
For $1 \le i \le \lceil \ln n_1 \rceil$ define the function $f_i: A \cup B \rightarrow R$ as follows:

\begin {eqnarray*}
\mbox { For vertices from $A$, }
   f_i(v_j)  &=& 0  \mbox { if the $i$th bit of $j$ is 0 } \\
   f_i(v_j)  &=& 2  \mbox { if the $i$th bit of $j$ is 1 }  \\
 \mbox { For vertices in $ u \in B$, }    f_i(u) &=& 1
\end {eqnarray*}

\noindent Let $I_i$ be the unit interval graph defined on the vertex set $A \cup B$ such that
two vertices $u$ and $v$ are adjacent if and only if $|f_i(u) - f_i(v)| \le 1$.

Now define $H_1 = \bigcap_{i=1}^{\lceil \ln n_1 \rceil} I_i$.  Thus we have
$cub(H_1) \le \lceil \ln n_1 \rceil$.
\end {defi}

\begin {defi}
\label {H2definition}
Let $B = \{u_1, u_2, \ldots,u_{n_2}\}$.
For $1 \le i \le \lceil \ln n_2 \rceil$ define the function $g_i: A \cup B \rightarrow R$ as follows:

\begin {eqnarray*}
 \mbox { For vertices from $B$, }
   g_i(u_j)  &=& 0  \mbox { if the $i$th bit of $j$ is 0 } \\
   g_i(u_j)  &=& 2  \mbox { if the $i$th bit of $j$ is 1 }  \\
 \mbox { For vertices in $ v \in A$, }  g_i(v) &=&  1
\end {eqnarray*}

\noindent Let $J_i$ be the unit interval graph defined on the vertex set $A \cup B$ such that
two vertices $u$ and $v$ are adjacent if and only if $|g_i(u) - g_i(v)| \le 1$.

Now define $H_2 = \bigcap_{i=1}^{\lceil \ln n_2 \rceil} J_i$.  Thus $cub(H_2) \le \lceil \ln n_2
\rceil$.
\end {defi}

\begin {lemma}
\label {H1lemma}
$H_1$ is a super graph of $G$ such that if $u,v \in A$, then $(u,v) \notin E(H_1)$.
\end {lemma}

\begin {proof}
It is easy to check that $I_i$ is a super graph of $G$ for each $i$.
Thus $H_1$ is clearly a super graph of $G$.
For $u,v \in A$, let $u = v_j$ and $v = v_k$ where $k \ne j$. Then clearly there exists a $t$,
$1 \le t \le \lceil \ln n_1 \rceil$ such that $j$ and $k$ differs in the
$t$th bit position. Now it is easy to verify that  $u$ and $v$ will not be adjacent
in $I_t$.  It follows  that for any pair $(u,v)$ where $u, v \in A$ there
exists $I_t$ such that $(u,v) \notin E(I_t)$. Then clearly $(u,v) \notin E(H_1)$ also.
Hence the Lemma.
\end {proof}

\begin {lemma}
\label {H2lemma}
$H_2$ is a super graph of $G$ such that if $u,v \in B$, then $(u,v) \notin E(H_2)$.
\end {lemma}

\begin {proof}
The proof is similar to that of the  Lemma \ref {H1lemma}.
\end {proof}

\begin{theo}
Given a bipartite graph $G = (A \cup B, E)$,  $cub(G) \leq
2(\Delta'+2)  \lceil \ln n_2 \rceil $.
\end{theo}

\begin {proof}
By Lemma \ref {3.2}, there exists a super graph $G^*$ of  $G$ such that if $u \in A$, $v\in B$
and $(u,v) \notin E(G)$, then $(u,v) \notin E(G^*)$.  Also let $H_1$ and $H_2$ be
the super graphs of $G$, from definitions \ref {H1definition}  and \ref {H2definition}
 respectively.  Now we claim
that $G = G^* \cap H_1 \cap H_2$. Cleary $G^* \cap H_1 \cap H_2$ is a super graph of
$G$, because each of them is a super graph of $G$.
 Now to see that $G^* \cap H_1 \cap H_2 = G$ we only need to prove that if
$(u,v) \notin G$, then $(u,v)$ is not an edge of at least one of these three graphs.
Now, if $u \in A$ and $v \in B$, $(u,v) \notin E(G^*)$ by Lemma \ref {3.2}.  If $u,v \in A$,
then $(u,v) \notin E(H_1)$ by Lemma \ref {H1lemma} and if $u,v \in B$, then $(u,v) \notin E(H_2)$
by Lemma \ref {H2lemma}.

Now, $cub(G) = cub(G^* \cap H_1 \cap H_2) \le  cub(G^*)  + cub(H_1) + cub(H_2)$.
By Lemma \ref {3.2}  $cub(G^*) \le 2(\Delta' + 1) \ln n_2$. Also by the definition
 of $H_1$ and $H_2$ we have $cub(H_1) \le \lceil \ln n_1 \rceil$ and
$cub(H_2) \le \lceil \ln n_2 \rceil$.
Thus we have,
\begin{eqnarray*}
  cub(G) &\leq& 2(\Delta' + 1)\ln n_2 + \lceil \ln n_1 \rceil  + \lceil \ln n_2 \rceil  \\
   &\leq&  2(\Delta' + 1)\ln n_2 + 2 \lceil \ln n_2 \rceil  \ \ \ \ \ \ \mbox{as $n_1 \leq n_2$}\\
   &=& 2(\Delta' + 2) \lceil \ln n_2 \rceil
\end{eqnarray*}

Hence the theorem.
\end{proof}

\noindent {\bf Remark:} In view of the Remark after Lemma \ref
{3.2}, we can infer that if $t \geq 3(\Delta'+1)\ln n_2$,   $G =G^*
\cap H_1 \cap H_2$ with high probability. But then the cube
representation output by the algorithm will be in $3(\Delta'+1) \ln
n_2 +  \lceil \ln n_2 \rceil + \lceil \ln n_1 \rceil   \leq
 3(\Delta'+2) \lceil \ln n_2 \rceil $ dimensions.
The following Theorem   gives the time complexity of our randomized algorithm to
construct such a cube representation.

\begin{theo}
Let $G = (A \cup B, E)$ be a bipartite graph with $n = n_1 + n_2$
vertices, $m$ edges and let $\Delta' = \min\{\Delta_A, \Delta_B\}$.
Then, with high probability, the cube representation of $G$ in
$3(\Delta'+2) \lceil \ln n_2 \rceil$ dimensions can be generated in
$O(\Delta'(m+n)\ln n_2)$ time.
\end{theo}

\begin{proof}
We assume that a random permutation $\pi$ on $n_1$ vertices can be
computed in $O(n_1)$ time. Recall that we assign $n$ intervals to
$n$ vertices as follows. If $v \in A$, then we assign the interval
$[\pi(v), n+\pi(v)]$ to $v$. If $v \in B$, then let $t = \min_{x \in
N(v)}{\pi(x)}$. Now, the interval $[t+n, t+2n]$ is given to the
vertex $v$. Since number of edges in the graph $m =
\frac{1}{2}\sum_{u \in A\cup B}{d(u)}$, one invocation of {\bf
RANDUNIT($G$)} needs $O(m+n)$ time. Since we need to invoke the
algorithm {\bf RANDUNIT($G$)} $O(\Delta'\ln n_2)$ times, the overall
algorithm that generates the cube representation in $3(\Delta'+2)
\lceil \ln n_2 \rceil $ dimensions runs in $O(\Delta'(m+n)\ln n_2)$
time
\end{proof}


\end{document}